\documentclass[final,12pt,notitlepage]{article}
\usepackage{amsfonts}
\usepackage{amsmath}
\usepackage{mathtools}
\usepackage{amsthm}
\usepackage{amssymb}
\usepackage[svgnames]{xcolor}
\usepackage[sc,labelsep=period,textfont=normalfont]{caption}

\usepackage[margin=1.2in]{geometry}
\usepackage[bottom,multiple,splitrule]{footmisc}
\usepackage{graphicx}
\usepackage[authoryear,round,longnamesfirst]{natbib}
\usepackage[inline,shortlabels]{enumitem}
\setlist[enumerate]{topsep=4pt,itemsep=4pt,partopsep=1ex,parsep=0pt}
\setlist[enumerate,1]{label=\emph{(\alph*)}}
\setlist[enumerate,2]{label=\emph{(\roman*)}}
\setlist[enumerate,3]{label=\emph{(\arabic*)}}
\setlist[itemize]{topsep=4pt,itemsep=4pt,partopsep=1ex,parsep=0pt}
\usepackage{hyperref}
\usepackage[onehalfspacing]{setspace}
\usepackage{titlesec}
\usepackage{sectsty}
\usepackage{dsfont}
\usepackage[capitalize,nameinlink]{cleveref}
\usepackage{xargs}
\usepackage{comment}
\usepackage{booktabs}
\usepackage[font={footnotesize},justification=justified]{caption}
\usepackage[labelformat=simple]{subcaption}
\usepackage{multirow,array}
\usepackage{tkz-euclide}
\usetikzlibrary{angles} 
\usepackage{tikz}
\usetikzlibrary{patterns, decorations.pathreplacing, arrows.meta,automata,positioning}
\usetikzlibrary{arrows}
\usetikzlibrary{calc}

\hypersetup{
pdftitle={Coalitions in Repeated Games},
pdfauthor={S. Nageeb Ali and Ce Liu},
pdfkeywords={self-enforcing, repeated game, cooperative games}
}
\hypersetup{
colorlinks,breaklinks,
citecolor=DarkBlue,urlcolor=Indigo,linkcolor=DarkBlue
}

\bibpunct[, ]{(}{)}{;}{a}{}{,}
\theoremstyle{plain}
\newtheorem{theorem}{Theorem}

\newtheorem{proposition}{Proposition}
\newtheorem{definition}{Definition}

\newtheorem{lemma}{Lemma}

\theoremstyle{definition}

\renewcommand{\tilde}{\widetilde}

  {\list{}{\leftmargin=0.2in\rightmargin=0.2in}\item[]}%
  {\endlist}
  {\list{}{\leftmargin=0.3in\rightmargin=0.3in}\item[]}%
  {\endlist}

\renewcommand{\paragraph}[1]{{\flushleft \textbf{\color{DarkRed}#1 }}}

\newcommand{\coalitions}{\mathcal{C}}

\renewcommand{\Re}{\mathbb{R}}
\newcommand{\hchi}{\raise0.4ex\hbox{$\chi$}}

\DeclareMathOperator*{\argmin}{arg\,min}

\renewcommand{\equiv}{:=}

\sectionfont{\color{DarkRed}}
\subsectionfont{\color{DarkRed}}
\subsubsectionfont{\color{DarkRed}}

\begin{document}


\title{On Conservative Stable Standard of Behavior and Perfect Coalitional Equilibrium\thanks{We used \href{refine.ink}{Refine.ink} for proofreading.}
}

\author{S. Nageeb Ali\thanks{Department of Economics, Pennsylvania State University. Email: \href{mailto:nageeb@psu.edu}{nageeb@psu.edu}.}\and Ce Liu\thanks{Department of Economics, Michigan State University. Email: \href{mailto:celiu@msu.edu}{celiu@msu.edu}.} }

\date{April 6, 2026}

\maketitle

\begin{abstract}

We show that in \cite{greenberg1989application}'s coalitional repeated game situation, every nondiscriminating Conservative Stable Standard of Behavior is a subset of the set of Perfect Coalitional Equilibrium \citep{aliliu2026main} paths. Moreover, the set of Perfect Coalitional Equilibrium paths itself is a nondiscriminating Conservative Stable Standard of Behavior. The set of Perfect Coalitional Equilibrium paths is therefore the maximal nondiscriminating Conservative Stable Standard of Behavior.
\end{abstract}

\newpage



 \setcounter{page}{1}
\setstretch{1.25}

\section{Overview}
\cite{greenberg1989application} applied the theory of social situations to infinitely repeated strategic form games with discounting. The focus of \cite{greenberg1989application} is the \emph{repeated game Nash situations} $(\gamma,\Gamma)$, defined on page 284-285 of the paper, where $\Gamma$ is the set of ``positions'' in a repeated strategic form game, and  $\gamma$ is the inducement correspondence that allows only \textit{individual} deviation. The central result of \cite{greenberg1989application}, his Theorem 6.2, demonstrates that in a repeated game Nash situation $(\gamma,\Gamma)$, his solution concept---\emph{Conservative Stable Standard of Behavior (CSSB)}---is equivalent to subgame perfect Nash equilibrium (SPNE) in the following sense:
\begin{theorem}(\citealp[Theorem 6.2]{greenberg1989application}) \label{thm:greenberg-main}
Let $PEP$ denote the collection of equilibrium paths in all SPNE, and let $\sigma^P$ be the standard of behavior (SB) defined by $\sigma^P(G)=PEP$  {for all } $G\in \Gamma$. Then $\sigma^P$ is the unique maximal nondiscriminating CSSB for the repeated Nash game situation $(\gamma,\Gamma)$.
\end{theorem}

In Section 7 of \cite{greenberg1989application}, he proposes a modification of the repeated game Nash situation $(\gamma, \Gamma)$ to allow for coalitional deviations in repeated strategic form games. More specifically, at the bottom of page 291, \citeauthor{greenberg1989application} proposes modifying the individual deviation inducement correspondence $\gamma$ to allow coalitional deviations (denoted $\gamma^{\coalitions}$), which results in the \emph{coalitional repeated game situation} $(\gamma^{\coalitions},\Gamma)$. His specification here mirrors that of Example 1 of \cite{aliliu2026main}. In light of the connection between SPNE and CSSB for the repeated game Nash situation $(\gamma,\Gamma)$ demonstrated in his Theorem 6.2 (\cref{thm:greenberg-main} quoted above), \citeauthor{greenberg1989application} suggests studying the maximal nondiscriminating CSSB in the coalitional repeated game situation $(\gamma^{\coalitions},\Gamma)$. While he does not prove any results, he observes on p. 292 (Example 6.3) that in a particular common interest game, there is a unique nondiscriminating CSSB in $(\gamma^{\coalitions},\Gamma)$ corresponding to  the unique Pareto-optimal action profile.

His observation suggests a connection to Theorem 1 in \cite{aliliu2026main}: as we illustrate in Table 3(A) of that paper, every Perfect Coalitional Equilibrium (PCE) plays the Pareto-optimal action profile in common interest games. 

In this note, we establish a general connection between nondiscriminating CSSB in \citeauthor{greenberg1989application}'s {coalitional repeated game situation} $(\gamma^{\coalitions},\Gamma)$ and our solution concept, PCE. We prove that for repeated strategic form games, the relationship between PCE and  nondiscriminating CSSB in the coalitional repeated game $(\gamma^{\coalitions},\Gamma)$ precisely parallels that between SPNE and nondiscriminating CSSB in the ``Nash repeated'' game $(\gamma,\Gamma)$. We prove the following result as the coalitional analogue of  \cref{thm:greenberg-main} quoted above.

\begin{theorem} \label{thm:greenberg-main-analogue}
Let $PCEP$ denote the collection of equilibrium paths in all PCEs, and let $\sigma^{PC}$ be the SB defined by $\sigma^{PC}(G)=PCEP$  {for all } $G\in \Gamma$. Then $\sigma^{PC}$ is the unique maximal nondiscriminating CSSB for the coalitional repeated game situation $(\gamma^{\coalitions},\Gamma)$.
\end{theorem}

To prove \cref{thm:greenberg-main-analogue}, we first generalize \citeauthor{greenberg1989application}'s {Propositions 5.3 and 5.4} from his repeated game Nash situation $(\gamma,\Gamma)$ to his coalitional repeated game situation $(\gamma^{\coalitions},\Gamma)$.  Next, we turn to PCE and show that the set of equilibrium paths is compact and extend \cite{abreu1988opc}'s optimal penal code for SPNE to PCE; neither of these results is contained in \cite{aliliu2026main}. Using these steps, we then prove \cref{thm:greenberg-main-analogue}. 
\medskip

In \cref{section: nash situation,section: coalition situation}, we offer a detailed account of \citeauthor{greenberg1989application}'s repeated game Nash situation $(\gamma,\Gamma)$ and coalitional repeated game situation $(\gamma^{\coalitions},\Gamma)$. \Cref{section: proof main} proves  \cref{thm:greenberg-main-analogue}. \Cref{section: proof supp} presents the intermediate results together with their proofs. Throughout, we stay as close as possible to \citeauthor{greenberg1989application}'s notation, making only minimal changes when they improve exposition.

\section{Repeated Game Nash Situation} \label{section: nash situation}
Let $N$ be a finite set of players, $Z=\times_{i\in N}Z^i$ be the set of stage-game action profiles, $\{u_i\}_{i\in N}$ be the stage game payoff functions, and $X = Z^{\infty}$ denote the set of feasible paths in the repeated game. A  {repeated game Nash situation} is a pair $(\gamma,\Gamma)$, where $\Gamma$ is the set of \emph{positions} describing continuation situations (e.g.\ a subgame or history), and $\gamma$ is the \emph{inducement correspondence}: for each set of players $S\subseteq N$, position $G\in \Gamma$ and continuation path $x\in X$, $\gamma(S|G,x)$ denotes the positions coalition $S$ can induced from the position $G$ when the continuation path $x$ is offered.

Formally, let $Z^t$ denote the set of $t$-tuples of action profiles, then the set of positions $\Gamma \equiv \cup_{t=0}^\infty Z^t$ is the set of histories. For each position $G= (z_1,z_2,\ldots,z_t)\in \Gamma$, define $a_i(G) = (1-\delta)\sum_{\tau=1}^t \delta^{\tau-1}u_i(z_{\tau}) $ and $b(G) =  \delta^t$. Let $U_i:X\rightarrow \Re$ denote player $i$'s discounted payoff function given by $U_i(x)= (1-\delta)\sum_{\tau=1}^{\infty} \delta^{\tau-1} u_i(x_{\tau})$ for all $x = (x_1,x_2,\ldots)\in X= Z^\infty$. The set of positions $\Gamma$ then satisfies that for every $G\in \Gamma$,
\[
 N(G) = N;\quad X(G) = X; \; \text{ and } u_i(G) = a_i(G) + b(G) U_i \text{ for all } i\in N.
\]
In the expressions above, $N(G)$ is the set of players at position $G$, $X(G)$ is the set of feasible outcomes at position $G$, and $u_i(G):X\rightarrow \Re$ is player $i$'s payoff function over outcomes at position $G$. In a repeated game Nash situation, the set of players and feasible outcomes are constant across all positions; in \citeauthor{greenberg1989application}'s language,  player's payoff from continuation path $x\in X$ at position $G= (z_1,z_2,\ldots,z_t)$ is given by $u_i(G)(x) = a_i(G) + b(G) U_i(x)$; we will follow this notation in this note.

In a repeated game Nash situation, the inducement correspondence $\gamma$ allows only individual deviations and coalitions are not allowed to form. For each positive integer $\tau$, position $G\in \Gamma$, path $x\in X$, player $i\in N$, and player $i$'s action in period $\tau$, $\zeta^i_\tau \in Z^i$, let $H=(G|x;\zeta^i_{\tau})\in \Gamma$ denote the position $H = (G,x_1,x_2,\ldots, x_{\tau-1},\zeta_\tau)$, where $\zeta_{\tau} \in Z^N$ is the action profile satisfying $\zeta^j_{\tau} = x^j_{\tau}$ for all $j\ne i$. For every $G\in \Gamma$ and $x\in X$, the inducement correspondence satisfies
\begin{align*} 
\gamma(\{i\} | G, x)= & \left\{\left(G | x ; \zeta_\tau^i\right) | \tau \in\{1,2, \ldots\} \text { and } \zeta_\tau^i \in Z^i\right\} \text{ for all } i\in N \\ 
& \gamma(S | G, x)=\varnothing \quad \text { for all }  S\subseteq N \text{ with } \quad|S|>1.
\end{align*}

A \emph{standard of behavior} (SB) is a mapping $\sigma: \Gamma \to 2^X$ that assigns to each position $G$ a set of continuation paths $\sigma(G) \subseteq X$. Given an SB $\sigma$ and a position $G$, a path $x\in X$ is said to be \emph{conservatively dominated} at $G$ (relative to $\sigma$) if there exists a coalition  $S\subseteq N$ and $H\in \gamma(S|G,x)$ such that $\sigma(H)\neq \varnothing$ and $u_i(H)(y) > u_i(G)(x)$ for all $i\in S$ and $y\in \sigma(H)$. Note that in the repeated game Nash situation, since the inducement correspondence $\gamma$ allows only individual deviations, a path $x\in X$ is {conservatively dominated} if and only if there exists player $i\in N$ and $H\in \gamma(\{i\}|G,x)$ such that $\sigma(H)\neq \varnothing$ and $u_i(H)(y) > u_i(G)(x)$ for all  $y\in \sigma(H)$.

Given an SB $\sigma$ and a position $G$, the \emph{conservative dominion} of $G$ relative to $\sigma$, denoted $CDOM(\sigma,G)$,  is the set of paths that are conservatively dominated at position $G$ (relative to $\sigma$). An SB $\sigma$ is:
\begin{itemize}
\item \emph{conservatively internally stable} if $\sigma(G) \cap CDOM(\sigma,G) = \varnothing$ for all  $G\in \Gamma$;
\item \emph{conservatively externally stable} if $X\backslash \sigma(G) \subseteq CDOM(\sigma,G)$
 for all  $G\in \Gamma$.
\end{itemize}
A {Conservative Stable Standard of Behavior}, or CSSB, is an SB that is both conservatively
internally and externally stable. An SB $\sigma$ is \emph{nondiscriminating} if $\sigma(G)=\sigma(H)$ for all  $G,H\in \Gamma$.

\section{Coalitional Repeated Game Situation} \label{section: coalition situation}

A \emph{coalitional repeated game situation} is a pair $(\gamma^{\coalitions},\Gamma)$ where $\Gamma$ is defined identically to the repeated game Nash situation in \cref{section: nash situation}, but the correspondence $\gamma^{\coalitions}$ allows not only individual but also coalitional deviations. In particular, for each positive integer $\tau$, $G\in \Gamma$, $x\in X$, $C \subseteq N$, and $\zeta^C_\tau \in Z^C$, let $H=(G|x;\zeta^C_{\tau})\in \Gamma$ denote the position $H = (G,x_1,x_2,\ldots, x_{\tau-1},\zeta_\tau)$, where $\zeta_{\tau} \in Z^N$ is the action profile satisfying $\zeta^j_{\tau} = x^j_{\tau}$ for all $j\notin C$. The inducement correspondence $\gamma^{\coalitions}$ is given by: for every $G\in \Gamma$, $x\in X$, and $C\subseteq N$,
\[
\gamma^{\coalitions}(C | G, x)=\left\{(G| x ; \zeta_\tau^C) : \tau \in\{1,2, \ldots\} \text{ and } \zeta_\tau^C \in Z^C\right\}.
\]

Given an SB $\sigma$, a path $x\in X$ is conservatively dominated at $G\in \Gamma$ if there exists a coalition $C\subseteq N$ and $H\in \gamma^{\coalitions}(C|G, x)$ such that $\sigma(H)\neq \varnothing$ and  $u_i(H)(y) > u_i(G)(x)$ for all $i\in C$ and $y\in \sigma(H)$. That is, a deviation is profitable for coalition $C$ if every member of $C$ strictly prefers every continuation at $\sigma(H)$ to $x$. By replacing $\gamma$ with $\gamma^{\coalitions},$ the conservative dominion of $G$ relative to $\sigma$, $CDOM(\sigma,G)$, can now be written as
\begin{align*}
 CDOM(\sigma,G) =  \Bigl\{& x\in X: \exists C\subseteq N, \tau\ge 1, \zeta^C_\tau\in Z^C
\text{such that } H=(G|x; \zeta_\tau^C) \text{ satisfies} \\
& \sigma(H)\neq\varnothing \text{ and } u_i(H)(y)>u_i(G)(x) \text{ for all } y\in \sigma(H) \text{ and } i\in C \Bigr\}.
\end{align*}
The notions of  conservative internal stability, conservative external stability, and CSSB extend directly from the repeated game Nash situation to the coalitional repeated game situation by adopting the definition of $CDOM$ for $\gamma^{\coalitions}$ instead of $\gamma$.

\section{Proof of \cref{thm:greenberg-main-analogue}} \label{section: proof main}
\subsection{Intermediate Results}
The proof of \cref{thm:greenberg-main-analogue} relies on four intermediate results. For convenience, we state them here without proof and defer their proofs to \cref{section: proof supp}. 
\medskip

Let $G^0$ be the position associated with the empty history, i.e. at the beginning of the game. The first result shows that the closure of a nondiscriminating CSSB in the coalitional repeated game situation is also a CSSB.

\begin{proposition} \label{prop:5-3}
Let $\sigma$ be a nondiscriminating conservative internally stable SB for the coalitional repeated game situation $(\gamma^{\coalitions}, \Gamma)$. For each $G\in \Gamma$, let $\tilde{\sigma}(G)$ be the closure of $\sigma(G)$. Then, the SB $\tilde{\sigma}$ is also a nondiscriminating conservative internally stable SB for $(\gamma^{\coalitions}, \Gamma)$. In addition, if $\sigma(G^0)\neq\varnothing$, then $\tilde{\sigma}(G^0)$ is also nonempty and compact.
\end{proposition}

Let $x, y \in X$, let $\tau$ be a positive integer, let $i \in N$, and let $\zeta_\tau^C \in Z^C$. Define
$$
\left(x ; \zeta_\tau^C ; y\right) \equiv \left(x_1, x_2, \ldots, x_{\tau-1}, \zeta_\tau, y_1, y_2, \ldots\right),
$$
where $\zeta_\tau^j = x_\tau^j$ for all $j \notin C$. This denotes the path obtained by following $x$ through period $\tau-1$, then having the players in $C$ deviate in period $\tau$ to $\zeta_\tau^C$, and then following the path $y$ thereafter. The next result provides an optimal-penal-code-like characterization of nondiscriminating CSSB.

\begin{proposition}
\label{prop:5-4}
Let $\sigma$ be a nondiscriminating SB for the coalitional repeated game
situation $(\gamma^{\coalitions},\Gamma)$, and suppose that for every
$i\in N$ the set $\argmin\{U_i(x):x\in \sigma(G^0)\}$ is nonempty. For each $i\in N$, choose
\[
z(i| \sigma)\in \argmin\{U_i(x):x\in \sigma(G^0)\}.
\]
Then a path $y\notin CDOM(\sigma,G^0)$ if and only if for every nonempty coalition $C\subseteq N$, every period $\tau\ge 1$, and every $\zeta^C_\tau\in Z^C$, there exists some player $i\in C$ such that $U_i(y)\ge U_i\bigl(y;\zeta^C_\tau; z(i| \sigma)\bigr)$.
\end{proposition}
\medskip

The next result turns to PCE and shows that the set of equilibrium paths is compact in the product topology over $X = Z^{\infty}$.

\begin{proposition}\label{thm:compact}
The set $PCEP$ is compact.
\end{proposition}

The next and final intermediate result obtains an optimal penal code for PCE, analogous to \cite{abreu1988opc}'s finding for subgame perfect Nash equilibria.

\begin{proposition} \label{thm:6-1}
A path $x[0] \in X$ belongs to $PCEP$ if and only if there exists a family of paths
$\{x[i]\}_{i\in N} \subseteq X$ such that for all $k\in \{0\}\cup N$, all
$\tau \ge 1$, all coalitions $C \subseteq N$, and all $\zeta^C_{\tau}\in Z^C$,
there exists $j \in C$ such that $U_j(x[k]) \ge U_j(x[k]; \zeta^C_{\tau}; x[j])$.
\end{proposition}

\subsection{Using the Intermediate Results to Prove \Cref{{thm:greenberg-main-analogue}}}
Since $\sigma^{PC}(G)=PCEP$ for all $G\in \Gamma$, $\sigma^{PC}$ is a nondiscriminating SB. We first prove that $\sigma^{PC}$ is a CSSB by proving conservative internal and external stability. 

For internal stability, fix $x[0]\in \sigma^{PC}(G^0)=PCEP$. Since $PCEP$ is compact by \cref{thm:compact} and each $U_i$ is continuous in the product topology, for every $i\in N$ there exists $z(i|\sigma^{PC})\in \argmin\{U_i(x):x\in PCEP\}$. By \cref{thm:6-1}, there exists a family of paths $\{x[i]\}_{i\in N}\subseteq PCEP$ such that for all $\tau\ge 1$, all coalitions $C\subseteq N$, and all $\zeta_\tau^C\in Z^C$, there exists $j\in C$ such that $U_j(x[0])\ge U_j(x[0];\zeta_\tau^C;x[j])$. Since by definition $U_j(x[0];\zeta_\tau^C;z(j|\sigma^{PC})) \le U_j(x[0];\zeta_\tau^C;x[j])$, we have
\begin{equation}\label{eq:1}
U_j(x[0])\ge U_j(x[0];\zeta_\tau^C;z(j | \sigma^{PC})).    
\end{equation}
By \cref{prop:5-4}, this implies $x[0]\notin CDOM(\sigma^{PC},G^0)$. Since $x[0]\in \sigma^{PC}(G^0)$ was arbitrary, $\sigma^{PC}$ is conservatively internally stable.

For conservative external stability, fix $x[0]\in X\backslash \sigma^{PC}(G^0)=X\backslash PCEP$. We will prove that $x[0]\in CDOM(\sigma^{PC},G^0)$. Suppose, toward a contradiction, that $x[0]\notin CDOM(\sigma^{PC},G^0)$. \cref{prop:5-4} then implies that for every nonempty coalition $C\subseteq N$, every period $\tau\ge 1$, and every $\zeta_\tau^C\in Z^C$, there exists $j\in C$ such that
\begin{equation}\label{eq:2}
    U_j(x[0])\ge U_j\bigl(x[0];\zeta_\tau^C;z(j|\sigma^{PC})\bigr).
\end{equation}
Now define $x[i]:=z(i|\sigma^{PC})$ for all $i\in N$, so $\{x[i]\}_{i\in N}\subseteq PCEP$. Using arguments similar to those leading to \eqref{eq:1}, this implies that for every $i\in N$, every nonempty coalition $C\subseteq N$, every period $\tau\ge 1$, and every $\zeta_\tau^C\in Z^C$, there exists $j\in C$ such that
\begin{equation}\label{eq:3}
U_j(x[i])\ge U_j\bigl(x[i];\zeta_\tau^C;z(j| \sigma^{PC})\bigr)
\end{equation}
Combining \eqref{eq:2} and \eqref{eq:3} and observing $z(j|\sigma^{PC}) = x[j]$, this shows that for all $k\in \{0\}\cup N$, all nonempty coalitions $C\subseteq N$, all periods $\tau\ge 1$, and all $\zeta_\tau^C\in Z^C$, there exists $j\in C$ such that $U_j(x[k])\ge U_j(x[k];\zeta_\tau^C;x[j])$.
By \cref{thm:6-1}, it follows that $x[0]\in PCEP$, contradicting the choice of $x[0]$. Therefore, $x[0]\in CDOM(\sigma^{PC},G^0)$ for every $x[0]\in X\backslash \sigma^{PC}(G^0)$. Hence $\sigma^{PC}$ is conservatively externally stable. Therefore $\sigma^{PC}$ is a nondiscriminating CSSB.

It remains to prove maximality and uniqueness. Let $\sigma$ be any nondiscriminating CSSB for $(\gamma^{\coalitions},\Gamma)$. Since every CSSB is in particular conservatively internally stable, \cref{prop:5-3} implies that its closure $\tilde{\sigma}$ is a nondiscriminating conservative internally stable SB. Because $\tilde{\sigma}(G^0)$ is compact, for each $i\in N$ choose
\[
\tilde{z}(i| \tilde{\sigma})\in \argmin\{U_i(x):x\in \tilde{\sigma}(G^0)\}.
\]

Since $\tilde{\sigma}$ is conservative internally stable, \cref{prop:5-4} implies that for every $x\in \tilde{\sigma}(G^0)$, every nonempty coalition $C\subseteq N$, every $\tau\ge 1$, and every $\zeta_\tau^C\in Z^C$, there exists $j\in C$ such that 
\[
U_j(x)\ge U_j(x;\zeta_\tau^C;\tilde{z}(j| \tilde{\sigma})).
\]
Fix any $x[0]\in \tilde{\sigma}(G^0)$ and define $\tilde{x}[i]:=\tilde{z}(i| \tilde{\sigma})$ for all  $i\in N$, then $\{\tilde{x}[k]\}_{k\in \{0\}\cup N}\subseteq \tilde{\sigma}(G^0)$. The preceding inequality shows that the hypothesis of \cref{thm:6-1} is satisfied, so $\tilde{x}[0]\in PCEP=\sigma^{PC}(G^0)$. Since $\tilde{x}[0]\in \tilde{\sigma}(G^0)$ was arbitrary, we have $\tilde{\sigma}(G^0)\subseteq \sigma^{PC}(G^0)$. Because $\sigma(G^0)\subseteq \tilde{\sigma}(G^0)$, it follows that $\sigma(G^0)\subseteq \sigma^{PC}(G^0)$ as well. Since both $\sigma$ and $\sigma^{PC}$ are nondiscriminating, the same inclusion holds for every $G\in \Gamma$. Thus $\sigma^{PC}$ contains every nondiscriminating CSSB, so $\sigma^{PC}$ is maximal. Finally, if $\sigma$ is any maximal nondiscriminating CSSB, then the same inclusions apply and maximality of $\sigma$ would imply that $\sigma=\sigma^{PC}$. Therefore $\sigma^{PC}$ is the unique maximal nondiscriminating CSSB. \qed

\section{Proof of Intermediate Results} \label{section: proof supp}

\subsection{Proof of \cref{prop:5-3}}
Since $\sigma$ is nondiscriminating, $\sigma(G)=\sigma(H)$ for all $G,H\in\Gamma$; taking closures yields $\tilde{\sigma}(G)=\tilde{\sigma}(H)$, so $\tilde{\sigma}$ is also nondiscriminating. If $\sigma(G^0)\neq\varnothing$, then $\tilde{\sigma}(G^0)$ is nonempty, and since $X$ is compact, it is compact as well.

It remains to show conservative internal stability. Fix $x\in \tilde{\sigma}(G^0)$ and suppose, toward a contradiction, that $x\in CDOM(\tilde{\sigma},G^0)$. Then there exist a coalition $C\subseteq N$, a period $\tau\ge 1$, and $\zeta^C_\tau\in Z^C$  such that, letting $H=(G^0| x;\zeta^C_\tau)$, we have $\tilde{\sigma}(H)\neq\varnothing$ and
\begin{equation}
u_i(H)(y)>u_i(G^0)(x)  \quad \text{for all } y\in \tilde{\sigma}(H)\text{ and }\ i\in C.
\label{1}
\end{equation}

Since $x\in \tilde{\sigma}(G^0)$, by the definition of closure, there exists a sequence $\big\{x[m]\big\}_{m=1}^\infty \in \sigma(G^0)$ with $x[m]\to x$ in the product topology. For each $m$, define
$H^m=(G^0 \,|\, x[m];\zeta_\tau^C)$, then $H^m\in \gamma^{\coalitions}(C | G^0,x[m])$ for all $m$. Since $\sigma$ is conservatively internally stable and $x[m]\in \sigma(G^0)$, we have that for all $m\ge 1$, $x[m]\notin CDOM(\sigma,G^0)$. Therefore by the definition of $CDOM(\sigma,G^0)$, for each $m$, there exist $y[m]\in \sigma(H^m)$ and $i_m\in C$ such that
\begin{equation}
u_{i_m}(H^m)(y[m])\le u_{i_m}(G^0)(x[m]).
\label{2}
\end{equation}

Since $C$ is finite, by passing to a subsequence we may assume that $i_m=i^* \in C$ for all $m$. Because $\sigma$ is nondiscriminating, $\sigma(H^m)=\sigma(H)$ for all $m$, so $y[m]\in \sigma(H) \subseteq \tilde{\sigma}(H)$ for all $m$. Since $\tilde{\sigma}(H)$ is compact, we may also assume (after passing to a further subsequence) that $y[m]\to y^*$ for some $y^*\in \tilde{\sigma}(H)$. By continuity,
\[
u_{i^*}(H^m)(y[m])\to u_{i^*}(H)(y^*),
\qquad
u_{i^*}(G^0)(x[m])\to u_{i^*}(G^0)(x).
\]
Taking limits in \eqref{2} yields $u_{i^*}(H)(y^*)\le u_{i^*}(G^0)(x)$, which contradicts \eqref{1} since $y^*\in \tilde{\sigma}(H)$ and $i^*\in C$. Thus $x\notin CDOM(\tilde{\sigma},G^0)$. Since $x\in \tilde{\sigma}(G^0)$ was arbitrary, $\tilde{\sigma}$ is conservatively internally stable. \qed

\subsection{Proof of \cref{prop:5-4}}
Fix $x\in X$. By definition, $x\in CDOM(\sigma,G^0)$ if and only if there exist a coalition $C\subseteq N$, a period $\tau\ge 1$, and $\zeta^C_\tau\in Z^C$ such that, letting $H=(G^0|x;\zeta_\tau^C)$, we have $\sigma(H)\neq\varnothing$ and $u_j(H)(y)>u_j(G^0)(x)$ for all $y\in \sigma(H)$ and $j\in C$.

Since $\sigma$ is nondiscriminating, $\sigma(H)=\sigma(G^0)$. Furthermore, since $z(i| \sigma)\in \argmin\{U_i(x):x\in \sigma(G^0)\}$, we have $x\in CDOM(\sigma,G^0)$ if and only if there exist a coalition $C\subseteq N$, a period $\tau\ge 1$, and $\zeta^C_\tau\in Z^C$ such that $u_j(H)\bigl(z(j |\sigma)\bigr)>u_j(G^0)(x)$ for all $j\in C$.

Negating this statement, we obtain that $x\notin CDOM(\sigma,G^0)$ if and only if for every coalition $C\subseteq N$, every period $\tau\ge 1$, and every $\zeta^C_\tau\in Z^C$, there exists some player $i\in C$ such that $u_i(H)\bigl(z(i | \sigma)\bigr)\le u_i(G^0)(x)$, where $H=(G^0|x;\zeta^C_\tau)$. 

By definition, $u_i(H)\bigl(z(i| \sigma)\bigr)=U_i(x;\zeta^C_\tau; z(i| \sigma))$ and $u_i(G^0)(x)=U_i(x)$. Substituting these identities yields that $x\notin CDOM(\sigma,G^0)$ if and only if for every coalition $C\subseteq N$, every period $\tau\ge 1$, and every $\zeta^C_\tau \in Z^C$, there exists some player $i\in C$ such that $U_i(x)\ge U_i(x;\zeta^C_\tau; z(i | \sigma))$. \qed

\subsection{Proof of \cref{thm:compact}}

We make use of the self-generation arguments in \cite{abreu1990toward}. First, we define the self-generation operator on the space of continuation paths. Given a set of paths $Y\subseteq X$, we use $\Psi(Y)$ to denote the set of paths that can be enforced as on-path equilibrium behavior of a PCE using continuation paths in $Y$, and begin by proving a series of lemmas that are useful for establishing \cref{thm:compact}.

\begin{definition}
For any set $Y \subseteq X$, define
\begin{align*}
\Psi(Y)
 \equiv  \Big\{ x\in &X: \text{for every } C\subseteq N,\ \tau\ge 1,\text{ and } \zeta^C_\tau\in Z^C,\\
& \exists \big\{\underline{x}[i] \big\}_{i\in N}\subseteq Y  \text{ such that }  U_i(x)\ge U_i(x;\zeta^C_\tau;\underline{x}[i]) \text{ for some } i\in C 
\Big\}.
\end{align*}
\end{definition}

\begin{lemma}\label{lemma:monotone}
The operator $\Psi$ is monotone: if $Y\subseteq Y'$, then $\Psi(Y)\subseteq \Psi(Y')$.
\end{lemma}

\begin{proof}
Immediate from the definition.
\end{proof}

\begin{lemma} \label{lemma:compact}
If $Y\subseteq X$ is compact, then $\Psi(Y)$ is compact.
\end{lemma}

\begin{proof}
Since $\Psi(Y)\subseteq X$ and $X$ is compact by Tychonoff's Theorem, it suffices to show that $\Psi(Y)$ is closed. Let $(x[m])_{m=1}^\infty \subseteq \Psi(Y)$ be a sequence converging to $x\in X$ in the product topology. Take an arbitrary coalition $C\subseteq N$, arbitrary period $\tau\ge 1$, and arbitrary $\zeta^C_\tau\in Z^C$. For each $m$, because $x[m]\in \Psi(Y)$, there exists a family $\{\underline{x}[m,i]\}_{i\in N}\subseteq Y$ such that  there exists some player $i\in C$ such that
\begin{equation}
    U_i(x[m])\ge U_i(x[m];\zeta^C_\tau;\underline{x}[m,i]).
\label{eq1}
\end{equation}
Since $Y$ is compact and $N$ is finite, after passing to a subsequence if necessary,
we may assume that for every $i\in N$, $\underline{x}[m,i]\to \underline{x}[i]$ for some  $\underline{x}[i]\in Y$.

We claim that the family $\{\underline{x}[i]\}_{i\in N}$ can be used as punishments to enforce $x\in \Psi(Y)$ against deviation $\zeta^C_\tau\in Z^C$ by coalition $C$ in period $\tau$. To see why, note that
By \eqref{eq1}, for each $m$ there exists
$i_m\in C$ such that
\begin{equation}\label{eq2}
U_{i_m}(x[m])\ge U_{i_m}(x[m];\zeta^C_\tau;\underline{x}[m,i_m]).    
\end{equation}
Since $C$ is finite, by passing to a further subsequence we may assume that
$i_m=i^*$ for all $m$, for some fixed $i^*\in C$. By continuity of payoffs and the
convergences
\[
x[m]\to x,\qquad \underline{x}[m,i^*]\to \underline{x}[{i^*}],
\]
taking limits in \eqref{eq2} yields $U_{i^*}(x)\ge U_{i^*}(x;\zeta^C_\tau;\underline{x}[{i^*}])$. Since $i^*\in C$, this verifies the enforceability conditions for $x$ against the coalition $C$ and deviation $\zeta^C_\tau$ in period $\tau$. Because $C$, $\tau$, and $\zeta_\tau$ were arbitrary, we conclude that $x\in \Psi(Y)$. Hence $\Psi(Y)$ is closed, and therefore compact.
\end{proof}

\begin{lemma} \label{lemma:enforceability}
If $Y\subseteq X$ satisfies $Y\subseteq \Psi(Y)$, then $Y\subseteq PCEP$.
\end{lemma}

\begin{proof}
Fix $x\in Y$. Since $x\in \Psi(Y)$, for every coalition $C$, period $\tau$, and deviation $\zeta^C_{\tau}$ from $x$,
there exists a family of paths $\{\underline{x}[i]\}_{i\in N}\subseteq Y$ such that  at least one member of $C$ is deterred by the corresponding continuation path $\underline{x}[i]$.

Construct a plan as follows. On path, play $x$. After any history corresponding to a coalition deviation $\zeta^C_\tau$ by coalition $C$ in period $\tau$, choose one player $i\in C$ satisfying $U_i(x)\ge U_i(x;\zeta^C_\tau;\underline{x}[i])$, and continue from period $\tau+1$ onward according to the path $\underline{x}[i]$. Proceed recursively at every history.

By construction, after every deviation at every history, there is some member of the deviating coalition who does not strictly gain. Hence no coalition profitably blocks at any history. Thus the constructed plan is a PCE whose equilibrium path is $x$, so $x\in PCEP$.
\end{proof}

\begin{lemma}\label{lemma:enforceabilityPCEP}
$PCEP \subseteq \Psi(PCEP)$.
\end{lemma}

\begin{proof}
Suppose $x\in PCEP$ is the equilibrium path generated by the PCE $\sigma$. By the definition of a PCE, for each coalition $C$, period $\tau$, and deviation $\zeta^C_{\tau}$, there exists a family of paths $\{x[i]\}_{i\in N}$  generated from the continuation play of $\sigma$ such that some player $i^*\in C$ is deterred from this deviation; furthermore, since the continuation of a PCE at any history is again a PCE, $\{x[i]\}_{i\in N}\subseteq PCEP$, so $PCEP \subseteq \Psi(PCEP)$.
\end{proof}

\paragraph{Proof of \cref{thm:compact}} Define recursively $X^0\equiv X$, $X^{k+1}\equiv\Psi(X^k)$ \text{for } $k\ge 0$. By \cref{lemma:monotone} $\Psi$ is monotone, and since $X$ is compact, by \cref{lemma:compact} the sets $\{X^k\}_{k=0}^\infty$ form a
decreasing sequence of compact sets. Therefore
\[
X^\infty:=\bigcap_{k=0}^\infty X^k
\]
is compact. To show that $PCEP$ is compact, we will show that $PCEP = X^{\infty}$. 

By \cref{lemma:enforceabilityPCEP}, $PCEP \subseteq \Psi(PCEP)$. Since  $\Psi$ is monotone and $PCEP\subseteq X = X^0$, we have $PCEP \subseteq \Psi(X^0) =  X^1$. Iterating this argument yields $PCEP \subseteq X^k$ {for all } $k\ge 0$, so $PCEP \subseteq X^\infty$. It remains to show that $X^\infty \subseteq PCEP$. By \cref{lemma:enforceability}, it suffices to prove that $X^\infty \subseteq \Psi(X^\infty)$.

The claim is trivially true if $X^\infty =\varnothing$. Suppose $X^\infty \ne \varnothing$, and fix an arbitrary $x\in X^\infty$ and an arbitrary coalition $C \subseteq N$, period $\tau \ge 1$, and deviation $\zeta^C_\tau \in Z^C$. Note that $x\in X^{k+1}=\Psi(X^k)$ for every $k\ge 0$. So for each $k$, there exists a family $\{\underline{x}[k,i]\}_{i\in N}\subseteq X^k$ satisfying the enforceability condition of $\Psi(X^k)$ for $x$.

 Since $X$ is compact, taking subsequences if necessary, each $\underline{x}[k,i] $ converges to some $\underline{x}[i] \in X$. Consider any player $i \in N$ and any $K \ge 0$. For all $k \ge K$, since the sequence $\{X^k\}_{k=0}^\infty$ is decreasing, it follows that $\underline{x}[k,i] \in X^{k} \subseteq X^K$. Because $X^K$ is compact, it is closed, so taking limits yields $\underline{x}[i] \in X^K$. Since $K$ was arbitrary, we conclude that
\[
\underline{x}[i] \in \bigcap_{K=0}^\infty X^K = X^\infty \text{ for all } i\in N.
\]
We will show that $\{\underline{x}[i]\}_{i\in N} $ can enforce $x\in X^\infty$.

Since $N$ is finite, we can select a single subsequence $\{k_\ell\}_{\ell=1}^\infty$ along which $\underline{x}[k_\ell,i] \to \underline{x}[i]$ {for all } $i \in N$. By the enforceability of $x$ in $\Psi(X^{k_\ell})$, for each $\ell$ there exists $i_\ell \in C$ such that
\[
U_{i_\ell}(x) \ge U_{i_\ell}\bigl(x;\zeta^C_\tau;\underline{x}[k_\ell,i_\ell]\bigr).
\]
Since $C$ is finite, by passing to a further subsequence we may assume that $i_\ell = i^*$ for all $\ell$, for some fixed $i^* \in C$. By continuity of payoffs and the convergence $\underline{x}[k_\ell,i^*] \to \underline{x}[i^*]$, taking limits yields $U_{i^*}(x) \ge U_{i^*}\bigl(x;\zeta^C_\tau;\underline{x}[i^*]\bigr)$, so $x$ is enforceable against coalition $C$ deviating in period $\tau$ and choosing $\zeta_{\tau}$ using punishments $\{\underline{x}[i]\}_{i\in N} \subseteq X^\infty$. Since the choice of $x\in X$, $C\subseteq N$, $\tau\ge 1$, and $\zeta^C_{\tau}\in Z^C$ above is arbitrary, this verifies $X^\infty \subseteq \Psi(X^\infty)$, so $PCEP = X^\infty$ and therefore is compact.\qed

\subsection{Proof of \cref{thm:6-1}}
We first prove the ``only if'' direction. Take an arbitrary $x[0]\in PCEP$. Since $PCEP$ is compact and each $U_i$ is continuous on $X$, for every player $i\in N$ we can select $x[i]\in \argmin \{ U_i(x) : x\in PCEP\}$.

Since $\{x[k]\}_{k=0}^n \subseteq PCEP \subseteq \Psi(PCEP)$, for every $k\in \{0\}\cup N$, every $\tau\ge 1$, every coalition $C\subseteq N$,
and every $\zeta^C_\tau\in Z^C$, there exist some player $j\in C$ and some
$y\in PCEP$ such that
\[
U_j(x[k]) \ge U_j(x[k];\zeta^C_\tau;y).
\]
Because $x[j]$ minimizes player $j$'s payoff over $PCEP$, $U_j(x[j]) \le U_j(y)$. Hence replacing the continuation path $y$ by $x[j]$ can only lower player $j$'s
continuation payoff, so
\[
U_j(x[k];\zeta^C_\tau;x[j])
\le
U_j(x[k];\zeta^C_\tau;y).
\]
Combining the two inequalities above yields $U_j(x[k]) \ge U_j(x[k];\zeta^C_\tau;x[j])$. Since $j\in C$, this proves the ``only if'' direction.

We now prove the ``if'' direction. Suppose there exists a family of paths $\{x[i]\}_{i\in N}\subseteq X$ such that for all $k\in \{0\}\cup N$, all $\tau\ge 1$, all coalitions $C\subseteq N$, and all $\zeta^C_\tau\in Z^C$,
there exists $j \in C$ with
\begin{equation}
U_j(x[k]) \ge U_j(x[k];\zeta^C_\tau;x[j]).
\label{eq:OPC-PCEP}
\end{equation}
We construct a plan $f$ as follows. There are $n+1$ states, indexed by $\{0\}\cup N$. In state $k$, the prescribed path is $x[k]$. If in state $k$ a coalition $C$ deviates at period $\tau$ via $\zeta^C_\tau$, choose one player $i=i(k,\tau,C,\zeta^C_\tau)\in C$ satisfying \eqref{eq:OPC-PCEP}, and from the next period onward switch to the path $x[i]$. 

We claim that $f$ is a PCE. Consider any history at which the continuation path is along $x[k]$, $k\in \{0\}\cup N$. If coalition $C$ deviates by choosing some $\zeta^C\in Z^C$, then by construction the continuation path following the deviation is $x[j]$ for some $j\in C$ satisfying \eqref{eq:OPC-PCEP}. Thus player $j$ does not strictly gain and this deviation is not profitable. Because the same argument applies along every path $x[k]$, no coalition can profitably deviate after any history. Therefore $f$ is a PCE. Its equilibrium path is $x[0]$, so $x[0]\in PCEP$. \qed

{
	\addcontentsline{toc}{section}{References}
	\setlength{\bibsep}{0.25\baselineskip}
	\bibliographystyle{jpe}
	\bibliography{coalitions}

@article{greenberg1989application,
  title={An application of the theory of social situations to repeated games},
  author={Greenberg, Joseph},
  journal={Journal of Economic Theory},
  volume={49},
  number={2},
  pages={278--293},
  year={1989},
  publisher={Elsevier}
}

@unpublished{aliliu2026main,
	Author = {Ali, S. Nageeb and Ce Liu},
	Note = {Working Paper},
	Title = {Coalitions in Repeated Games},
	Year = {2026}
}

@article{abreu1990toward,
	Author = {Abreu, Dilip and Pearce, David and Stacchetti, Ennio},
	Journal = {Econometrica},
	Pages = {1041--1063},
	Publisher = {JSTOR},
	Title = {Toward A Theory of Discounted Repeated Games with Imperfect Monitoring},
	Year = {1990}}

@article{abreu1988opc,
	Author = {Abreu, Dilip},
	Journal = {Econometrica},
	Number = {2},
	Pages = {383--396},
	Publisher = {Citeseer},
	Title = {On the Theory of Infinitely Repeated Games with Discounting},
	Volume = {56},
	Year = {1988}}
}

\end{document}